\documentclass[%
 reprint,
 amsmath,amssymb,longbibliography,aps,
]{revtex4-1}

\usepackage{natbib}
\usepackage{graphicx}
\usepackage{dcolumn}
\usepackage{bm}

\usepackage{graphicx}
\usepackage{dcolumn}
\usepackage{bm}
\usepackage{hyperref}
\usepackage{xcolor}
\usepackage{amsmath, amssymb, amsfonts}
\usepackage{mathtools}
\usepackage{subfigure}
\usepackage{mathrsfs}
\usepackage{comment}
\usepackage{amsthm}
\usepackage{ dsfont }

\newcommand{\col}[1]{\textcolor{red}{#1}}

\definecolor{purple1}{rgb}{128,0,128}

\newcommand{\bea}{\begin{eqnarray}}
\newcommand{\ea}{\end{eqnarray}}

\theoremstyle{definition}

\newtheorem{theorem}{Theorem}
\newtheorem{proposition}{Proposition}

\definecolor{dfcol}{cmyk}{1, 0.2108, 0.13, 0.3}
\newcommand{\df}[1]{\ifthenelse{\boolean{}}{\textcolor{dfcol}{[{\bf DF}: #1]}}{}}

\renewcommand{\[}{\begin{equation}}
\renewcommand{\]}{\end{equation}}

\newcommand{\ket}[1]{|#1\rangle}
\newcommand{\bra}[1]{\langle#1|}
\newcommand{\braket}[2]{\langle#1|#2\rangle}
\newcommand{\pro}[2]{|#1\rangle\langle#2|}
\newcommand{\mean}[1]{\langle#1\rangle}

\newcommand{\tr}{\mathrm{tr}}
\newcommand{\norm}[1]{\left\|#1\right\|}

\newcommand{\R}{{\hat{\rho}}}

\definecolor{mygray}{gray}{0.6}

\newcommand{\hamz}{\hat{H}}
\newcommand{\hamd}{\hat{V}}

\begin{document}


\title{Maximum quantum battery charging power is not an entanglement monotone}
\title{Beneficial and detrimental entanglement for quantum battery charging}

\author{Ju-Yeon Gyhm}
\affiliation{%
Seoul National University, Department of Physics and Astronomy, Center for Theoretical Physics, Seoul 08826, Korea
}%
\author{Uwe R. Fischer}%
\affiliation{%
Seoul National University, Department of Physics and Astronomy, Center for Theoretical Physics, Seoul 08826, Korea
}%

\date{\today}

\begin{abstract}
We establish a general implementation-independent approach to assess the 
potential advantage of using highly entangled quantum states 
between the initial and final states of the charging protocol 
to enhance the maximum charging power of quantum batteries. 
It is shown that the impact of entanglement on power can be separated 
from both the global quantum speed limit associated to an optimal choice of  
driving Hamiltonian and the energy gap of the batteries. 
We then demonstrate that the quantum state advantage of battery charging, defined as the power obtainable for given 
quantum speed limit and battery energy gap, is not an entanglement monotone.
A striking example we provide is that, 
counterintuitively, independent thermalization of the local batteries, completely destroying any entanglement, can lead to larger charging power than that of the initial maximally entangled state. 
Highly entangled states can thus also be potentially {\em disadvantageous} when compared to 
product states.  
We also demonstrate that taking the considerable effort of 
producing highly entangled states, such as W or $k$-locally entangled states, is not sufficient to obtain quantum-enhanced scaling behavior with the number of battery cells. 
{Finally, we perform an explicit computation for a Sachdev-Ye-Kitaev battery charger 
to demonstrate that the quantum state advantage allows the instantaneous power to exceed its 
classical bound.}
\end{abstract}

\maketitle


\section{Introduction} 

The modern notion of a ``quantum advantage" has its origin in  
the potential of quantum computers 
to outperform their classical counterparts \cite{Pan,Madsen}.
Other areas of information technology,  
for which a quantum advantage can be defined, are for example 
metrology~\cite{Lloyd_Science,Braun} 
which has inter alia been used in gravitational wave astronomy \cite{Abadie},  
or 
cryptography~\cite{RevModPhys.74.145}  
and 
communication~\cite{Gisin}. 

Where however quantum advantage seems to be most relevant for many practical reasons is on the field of energy storage and retrieval \cite{Auffeves}. 
Quantum battery research has focused on concrete means to 
increase charging power~\cite{Binder_2015,PhysRevLett.118.150601,Campaioli2018, PhysRevE.105.064119,Dario,PhysRevLett.128.140501} as well as on work extraction~\cite{Alicki,PhysRevLett.111.240401,Hamma,Barra,PhysRevLett.129.130602,PhysRevLett.130.210401}, see for a review including quantum thermodynamical aspects 
\cite{Dutta}. 
{In addition to these theoretical studies, increasingly progress has very recently been made towards an experimental realization of quantum batteries 
\cite{quachshort,Hu_2022,Joshi,Cruz}.} 

Quantum technology is generally driven by the increasing demand for a practically useful 
quantum advantage, 
which however can be an ambiguous and strongly model dependent notion \cite{Andolina}. 
{One may conjecture that 
entanglement is the most important resource for the 
quantum advantage of charging, see, e.g., Refs.~\cite{Ferraro,Salamon}, 
but results to the contrary have been found
~\cite{PhysRevE.102.052109,Le}.} 
As far as we are aware, there is however no general {device- and implementation-independent} relation which expresses the {potential} advantage of using entangled states 
solely in terms of an arbitrary quantum state specified by the density matrix operator $\R$.
Our major aim in this work is to fill this gap, such as to be able to quantify in a general manner how entanglement contributes to quantum advantage.

Physical mechanisms affecting 
the maximal charging power are the external driving for the charging, the 
general battery setup, and the energy gap between the ground state and the excited state of the batteries~\cite{Salamon, Campaioli2018}. 
We show in the following that the maximum charging power can be classified and separated into three parts, the energy gap of the batteries, the quantum speed limit associated to the time dependent  driving Hamiltonian, 
and {what one may call} a genuine quantum state advantage depending solely on the quantum state. 
We thereby provide a general relation rigorously establishing the relation of entanglement and quantum charging advantage, by {\em isolating} the impact of entanglement from the quantum speed limit and the energy gap. 
We note in this regard that the energy gap and quantum speed limit constitute ``trivial" advantages in the sense that one is  always able to obtain a higher charging power from the faster time evolution of the 
quantum state and from a large energy gap. Isolating these, we aim at finding 
the optimal charging path in the Hilbert space outperforming the classical parallel charging protocol. 
This can then represent a genuine quantum state advantage of charging, because this Hilbert space trajectory is determined by whether the states are entangled or not. 
We point out in this regard, and importantly for our argument, that the quantum speed limit (QSL) has no bias to be enhanced by entanglement in a given quantum state. Specifically, the QSL samples the full 
Hilbert space generated by arbitrarily choosing the driving Hamiltonians. It can therefore not depend on any kind of bi- or multi-partition of the system, and thus also not on any corresponding measure of entanglement. 

By isolating the quantum advantage related to entanglement, we obtain an explicit expression for the entanglement contribution to charging power,  
from which it {generally} follows that while entanglement is necessary for a 
quantum state advantage  over classical charging, it is not sufficient,  cf.~Ref.~\cite{Campaioli2018}.  
We thus establish the latter fact in a device- and implementation-independent way. 
We find that certain varieties of entanglement can {\em diminish} charging power and
prove, as a theorem, 
that  quantum state advantage  is not an entanglement monotone. 
We demonstrate this for a number of examples. Therefore, 
the quantum advantage of battery charging cannot constitute a measure of entanglement.

After introducing the general battery setup, we discuss the conventional approach for evaluating the power bound via the covariance matrix, which is however, as we show, hampered by containing also classical 
{in addition to quantum  correlations.} To resolve this, we introduce {what we coin a 
{\em commutation matrix,}} and from this establish the fundamental entanglement-determined bound on charging power.
{Our approach explicitly
demonstrates that it is not possible to assess the maximally possible
quantum advantage of battery charging solely 
with an entanglement measure. }
 This leads, inter alia, to the (at first glance counterintuitive) possibility that increasing charging power can 
 potentially be achieved  by destroying entanglement, e.g., by thermalization. We show this explicitly for highly  entangled initial states during their Lindbladian evolution to thermalized states. 
We also demonstrate, for globally entangled W and $k$-locally entangled pure states, that their quantum state advantage  
does not display a scaling with the number of cells $N$. Therefore, their obtainable power does not scale 
faster than linear in $N$.  
We perform furthermore an explicit computation for Sachdev-Ye-Kitaev batteries \cite{Dario} 
demonstrating that the quan{tum state advantage defined by us indeed peaks at 
approximately the same time as the instantaneous power does. 

 We provide our major results in the main text, and defer detailed derivations 
 and proofs to an 
  Appendix.

\section{Quantum battery power}
We consider a quantum battery setup composed of independent cells, 
with Hamiltonian 
$\hamz=\sum_{i=1}^N \hamz_i\bigotimes_{j\neq i}\hat{I}_j$, where $\hamz_i$ indicates the Hamiltonian of the 
$i$th cell. At time $t$ an instantaneous quantum state of the battery, represented by a density matrix 
$\R(t)$, 
evolves due to switching on for a finite time 
a time-dependent driving Hamiltonian, $\hamd (t)$ (which contains the 
time independent battery part $\hamz$), according to the von Neumann equation ($\hbar =1$)
\bea
\label{eq:state_evolution}
\frac{\partial\R}{\partial t} 
=i[\R,\hamd(t)] .
\ea
The energy stored in the battery is 
$E(t) = \tr(\hamz \R)$ and 
the instantaneous battery charging power is then defined as 
a change in the (expectation value of the) energy stored in the battery cells per unit of time, 
\bea
\label{eq:power}
P(t)=\tr\left(\hamz\,  
\frac{\partial \R}{\partial t} 
\right) 
=\tr\left(i\hamz[\R,\hamd(t)]\right). 
\ea
We aim at determining to which extent the entanglement of the state $\R$
 impacts {the upper bound on the power of battery charging.} This can be reformulated as determining the maximum instantaneous power of battery charging 
 we can obtain from the state $\R(t)$ by manipulating $\hamd(t)$. 
 
{The relation $P=\tr\left(i\hamz[\R,\hamd(t)]\right) = \tr\left(i[\hamz,\hamd(t)]\R\right)$ 
for the instantaneous power, defining
variances 
of operators by
$\mean{\Delta\boldsymbol{\cdot}^2}= \tr(\boldsymbol{\cdot}^2\R)-\tr(\boldsymbol{\cdot}\R)^2$, 
then reads by the Heisenberg-Robertson inequality} 
\bea\label{uncertanty}
P(t) 
\leq2\sqrt{\mean{\Delta \hamd^2(t)}\mean{\Delta \hamz^2}}. 
\ea
This implies  
that 
{the standard deviations}
of $\hamd(t)$ and $\hamz$ determine the bound of power~\cite{Salamon}, and 
leads one to conjecture that the \emph{covariance matrix} ~\cite{PhysRevLett.99.130504,li2008separability,gittsovich2010quantifying}, 
which generally relates to the variance 
of any given set of observables, can be applied to assess the 
quantum state advantage of battery charging.

\section{Covariance matrix approach}
\subsection{Definition of covariance matrix} 
We briefly review the definition of the covariance matrix of multipartite systems and its properties. 
{Covariance matrices, familiar from continuous variable systems (in particular for Gaussian states), were introduced by  
Refs.~\cite{PhysRevLett.99.130504,li2008separability,gittsovich2010quantifying},
to address the separability of finite-dimensional quantum states, and were demonstrated to provide a reasonably general framework to capture entanglement, in particular 
also by linking the covariance matrix criterion to previously established criteria for entanglement.}


To be able to define the covariance matrix, we need to first establish a set of {\em local observables}.  
We assume the $i$th cell lives in an $n_i$-dimensional Hilbert space.
The local observable set of the $i$th cell, $\hat{\bm{M}^i}$, has $n_i^2$ elements 
\[
\hat{\bm{M}^i}
=\{\bigotimes_{ j\neq i}\hat{I}^j \otimes A^i_1,\cdots,\bigotimes_{ j\neq i}\hat{I}^j \otimes A^i_{n_i^2}\}, 
\]
{where the $A_\alpha^i$ constitute the members of the Lie algebra of $U(n_i)$. }
Importantly, we impose the following Lie algebra orthonormality condition~\cite{PhysRevLett.99.130504},
\bea
{\tr(\hat{A^i_\alpha}\hat{A^i_\beta})=\delta_{\alpha\beta} \,\quad \forall\, i}
\label{Acond}
\ea
on the operators contained in the set of observables, where the $\alpha,\beta$ indices run over
$1,\ldots,n_i^2$. The orthonormality condition is imposed to be able to 
render the norm of the covariance matrices in Eq.~\eqref{CovDef} below 
independent of the observables set $\hat{\bm M}$, then leaving only state dependence of the operator norm.

The total set of observables  is the union of the local cell observables as follows 
\bea
\hat{\bm M}=\bigcup_{i=1}^N \hat{\bm M^i}=\{\hat{M}_\mu\}
\label{M}
\ea
The covariance matrix $\gamma(\R,\hat{\bm M})$ is then defined as
the symmetrized correlation function,
\bea
\gamma(\R,\hat{\bm M})_{\mu\nu}=\tfrac{1}{2}\mean{\hat{M_\mu} \hat{M_\nu}+\hat{M_\nu} \hat{M_\mu}}
-\mean{\hat{M_\mu}}\mean{\hat{M_\nu}}, \label{CovDef}
\ea
which is a symmetric, real, and positive semidefinite matrix. \
Here, indices $\mu,\nu$ 
run over both $i$ and $\alpha$ indices.

{We make use of the following properties of the 
 covariance matrix $\gamma$. First, the eigenvalues of $\gamma(\R,\hat{\bm M})$ are independent of the observable set $\hat{\bm M}$ when the orthonormality 
condition \eqref{Acond} is met. 
Another property  is that the variance 
of $\hamz$ obeys an inequality involving $\gamma$. 
To establish that inequality, 
we use that the Hamiltonian can be written as a linear decomposition using the operator set 
$\hat{\bm M}$ \cite{PhysRevLett.99.130504}. Employing a normalized real vector, ${\bm u}$, 
$\hamz=\hat{\bm  M}\cdot {\bm u}\sqrt{\sum_{i=1}^N\tr(\hamz_i^2)}$, we have  }
\[
\begin{split}
\mean{\Delta\hamz^2}&={\bm u}
^\mathrm{T}\gamma(\R,\hat{\bm M}){\bm u}
\sum_{i=1}^N\tr(\hamz_i^2)
\\&
\leq\|\gamma(\R,\hat{\bm M})\|\sum_{i=1}^N\tr(\hamz_i^2),
\end{split}
\]
from which we can infer the bound on the charging power by using \eqref{uncertanty},
\bea
\label{power_uncertanty}
P\leq 2\sqrt{\|\gamma(\R,\hat{\bm M})\|\sum_{i=1}^N\tr(\hamz_i^2)\mean{\Delta \hamd^2}},
\ea
where the operator norm $\|\boldsymbol{\cdot}\|$ of a Hermitian operator 
 is its largest eigenvalue \cite{Hassani}. 
 In addition, $\|\gamma(\R,\hat{\bm M})\|$ is bounded by $\frac{N}{2}$ and is equal to $\frac{1}{2}$ for a {product state,} $\ket{\psi}=\bigotimes_{i=1}^{N} \ket{\psi_i}$. Finally, $\|\gamma(\R,\hat{\bm M})\|$ is conserved by local unitary evolution~\cite{Salamon}.

\subsection{Issues with the covariance matrix approach}
From the fact that $\|\gamma(\R,\hat{\bm M})\|$ 
is invariant for any {pure product state,} one tends to infer that $\|\gamma(\R,\hat{\bm M})\|$ represents
{a suitable measure to assess} to which extent  entanglement contributes to the {power bound.} 
Indeed entanglement contributes for a pure state, but however not generally when one is dealing with a mixed state.

In particular, mixed states 
can increase $\|\gamma(\R,\hat{\bm M})\|$ although they do not enlarge the bound on power.
 To illustrate this with a concrete example, the 
mixed state $\R=(\ket{00}\bra{00}+\ket{11}\bra{11})/2$ has $\|\gamma(\R,\hat{\bm M})\|=1$ larger than $\frac12$ for {simple product states.} 
This mixed state thus seems to indicate a factor two advantage over {product states.} However, 
this is not the case, which can be seen as follows. 

By decomposing $\R$ into the normalized sum of $\R_1=\ket{00}\bra{00}$ and $\R_2=\ket{11}\bra{11}$, we rewrite Eq.~\eqref{uncertanty} as
\[
P\leq \frac{1}{2}(\tr([-i[\hamz,\hamd],\R_1])+\tr([-i[\hamz,\hamd],\R_2])). 
\]
{
This yields, using \eqref{power_uncertanty}, 
\[
    P\leq 
   \sum_{l=1,2} \sqrt{\|\gamma(\R_l,\hat{\bm M})\|
    \sum_{i=1}^N\tr(\hamz_i^2)\mean{\Delta\hamd^2}_l},
\]
where $\mean{\Delta\hamd^2}_l=\tr(\hamd^2\R_l)-\tr(\hamd\R)^2$. Since
\[
\frac12 \left(\sqrt{\mean{\Delta \hat V^2}_1}+\sqrt{\mean{\Delta \hat V^2}_2}\right)\leq \sqrt{\mean{\Delta \hat V^2}}
\]
and both $\|\gamma(\R_l,\hat{\bm M})\|$ 
are equal to $\frac12$, we obtain finally that
\[ 
    P\leq\sqrt{2\sum_{i=1}^N\tr(\hamz_i^2)\tr(\Delta \hamd^2)}, 
\]
{identical to the product state bound, thus leading to a tighter bound than that in \eqref{power_uncertanty} when one sets $\|\gamma(\R,\hat{\bm M})\|=1$ for the non-decomposed mixed state therein. This 
provides evidence that the covariance matrix $\gamma$ 
also contains classical correlations.}} 

We show below that, generally, there is no quantum state advantage from 
separable states over {product states.} 
{Mixed states can increase $\|\gamma(\R,\hat{\bm M})\|$,  however 
 inequality~\eqref{power_uncertanty} is thereby rendered a loose  bound 
 and the inequality on the power cannot be saturated.}


In summary, the inequality \eqref{power_uncertanty} only furnishes a loose, non-saturable bound on the 
charging power, i{n which, in particular, the impact of  a quantum speed limit 
{(which does not give a directly entanglement-related 
factor in the power bound)} 
is not manifest.
To address these issues, let us define the what we coin the \emph{commutation matrix}.

\section{commutation matrix} 
\subsection{Definition}
{To eliminate the above discussed impact of classical correlations,} 
we define the {\em commutation matrix} $\gamma_C(\R,\hat{\bm M})$ as follows 
\bea
\gamma_C(\R,\hat{\bm M})_{\mu\nu}\coloneqq-\tr([\hat{M_\mu},\sqrt{\R}][\hat{M_\nu},\sqrt{\R}]),
\label{defineGP}
\ea
for the same orthonormalized observable set $\hat{\bm M}$ 
displayed in Eq.~\eqref{M}, used already for the covariance matrix. Since $\R$ is positive semidefinite, $\sqrt{\R}$ is well-defined as a positive semidefinite matrix. The matrix 
$\gamma_C(\R,\hat{\bm M})$ has similar properties as $\gamma(\R,\hat{\bm M})$: Also 
$\gamma_C(\R,\hat{\bm M})$ is positive semidefinite and its eigenvalues 
only depend on the quantum state for all observables in the set 
which meet the orthonormality condition~\eqref{Acond}. 
We derive these and other properties of the commutation matrix in Appendix \ref{propcommmatrix}. 

\subsection{Power bound}
Crucial for our argument in the following is the property that  
$\|\gamma_C(\R,\hat{\bm M})\|$ is only a function of the quantum state encoded in the density matrix, $\R$.
We denote this function as 
\begin{equation}
\Gamma_C(\R)\coloneqq \|\gamma_C(\R,\hat{\bm M})\|. \label{defC}
\end{equation} 
{Using this definition of $\Gamma_C(\R)$ in terms of our commutation matrix, 
we now obtain a tighter bound than 
Eq.~\eqref{power_uncertanty}:} 
\begin{theorem}\label{theorem_main}
The instantaneous power of the quantum battery is bounded as follows 
\[\label{inequal_power}
P\leq \sqrt{2\kappa\Gamma_C(\R)\sum_{i=1}^N\tr(\hamz_i^2)\mean{\Delta\hamd^2}},
\] 
{where the coefficient $\kappa$ lies within the range $1\leq \kappa\leq 2$ (see Appendix \ref{proof_power}), 
and we neglect a possible weak dependence of $\kappa$ on $\R$ and $\hamd$ for mixed states; for any pure state,  $\kappa=1$.} 
Equality and thus saturation of the bound \eqref{inequal_power} 
is met when the two conditions 
\bea
&&\{\sqrt{\R},\hamd\}=i[\hamz,\sqrt{\R}]\sqrt{-\frac{\tr(\{\sqrt{\R},\hamd\}^2)}{\tr([\hamz,\sqrt{\R}]^2)}}
\label{eq_con_1}\\
&&\hamz=\hat{\bm M}\cdot {\bm u}
\sqrt{\sum_{i=1}^N\tr(\hamz_i^2)}\label{eq_con_2} 
\ea
are fulfilled.
Here, ${\bm u}
$ is the normalized eigenvector of $\gamma_C(\R,\hat{\bm M})$ which has the eigenvalue  $\Gamma_C(\R)$. The theorem on the power bound can be proved as follows. 
\end{theorem} 
\begin{proof}
By the Cauchy-Schwarz inequality, we have 
\begin{multline}
\label{power cauchy ineq}
P(t)=\tr(i[\hamz,\sqrt{\R}]\{\sqrt{\R},\hamd\})\\
\leq\sqrt{-\tr([\hamz,\sqrt{\R}]^2)\tr(\{\sqrt{\R},\hamd\}^2)}
\end{multline}
from which the first equality condition~\eqref{eq_con_1} derives. Moreover, we can always choose a $\hamd$ which satisfies \eqref{eq_con_1}. A detailed method to choose an optimal $\hamd$ is provided by Eq.~\eqref{eigen_v_ij}  
in Appendix \ref{proof_power}.


We now establish how $2\kappa\mean{\Delta\hamd^2}$ derives from $\tr(\{\sqrt{\R},\hamd\}^2)$.
{First of all, we obtain, by using the fact that $\Gamma_C$ is the largest eigenvalue 
of $\gamma_C(\R,\hat{\bm M})$ 
 [cf.~Eq.~\eqref{fyou}] 
\[-\tr([\hamz,\sqrt{\R}]^2)\leq\Gamma_C(\R) \sum_{i=1}^N\tr(\hamz_i^2)
\label{ineqGamma}. \]}\noindent Since adding real multiples of the identity operator to $\hamd$ does not contribute to the power, the bound on the power in Eq.~\eqref{power cauchy ineq} can be rewritten as follows (with real $\lambda$): 
\[\label{inequal_power_lambda}
    P\leq \min_\lambda \sqrt{\Gamma_C(\R)\sum_{i=1}^N\tr(\hamz_i^2)\tr(\{\sqrt{\R},\hamd+\lambda\hat{I}\}^2)}
\]
The term $\tr(\{\sqrt{\R},\hamd+\lambda\hat{I}\}^2)$ on the right-hand side of \eqref{inequal_power_lambda} is minimized by $2\tr(\R\hamd^2)+2\tr(\sqrt{\R}\hamd\sqrt{\R}\hamd)-4\tr(\R\hamd)^2$ when $\lambda$ is equal to $-\tr(\hamd \R)$. By the inequality $\tr(\R\hamd)^2\leq\tr(\sqrt{\R}\hamd\sqrt{\R}\hamd)\leq\tr(\R\hamd^2)$, we can replace $\tr(\{\sqrt{\R},\hamd\}^2)$ by $2\kappa\mean{\Delta\hamd^2}$, for $1\leq \kappa\leq 2$. 
The latter inequality is derived by using that 
$-\tr([\hamd,\sqrt{\R}]^2)\geq0$ and $\tr(\hamd\sqrt{\R}\hamd\sqrt{\R})\geq0$.  
This completes the proof of the inequality \eqref{inequal_power}; we provide further details in Appendix \ref{proof_power}. 
\end{proof}

\section{Properties of derived power bound}
 \subsection{Isolating the impact of quantum speed limit}
To compare driving Hamiltonians composed, e.g., only of local or global battery operations 
(or a combination of these), and their impact on the power bound, 
we  first explain the form of the QSL constraint we impose.
It directly derives from the time-energy uncertainty relation~\cite{Deffner_2017}.
{The standard deviation 
 of the driving Hamiltonian, 
 for all possible drivings $\hamd$, yields an 
 effective driving gap $\Delta E$. The latter is given by  
 the time average of the standard deviation of the charging operator as follows 
  (called the constraint $C_1$ in Ref.~\cite{PhysRevLett.118.150601}) 
\bea
\label{condition_QSL}
\Delta E = \int_0^T \frac{\sqrt{\mean{\Delta\hamd^2}}}{T} dt = \Delta E_{\mathrm{single}} 
\times \sqrt{N},
\ea 
where we assume that the single-cell gap  $\Delta E_{\mathrm{single}}$ 
does not depend on $\hamd$.
The above relation restricts the mean of the speed of time evolution of states in Hilbert space by 
imposing a finite gap $\Delta E$. }

We impose in the following the stronger constraint 
\bea
\sqrt{\mean{\Delta\hamd^2}}=\mathrm{constant}\times \sqrt{N}, 
\label{stronger_condition_QSL}
\ea 
which restricts the {\em instantaneous} rather than just the global speed of time evolution as specified by \eqref{condition_QSL}. 
  {The constant contained here equals $\Delta E_{\mathrm{single}} $}
contained in \eqref{condition_QSL} when the condition \eqref{stronger_condition_QSL} is imposed.

{Note the constraint \eqref{stronger_condition_QSL}
restricts the evolution speed of states in the Hilbert space {\em independent} from entanglement, due to the fact that the choice of the driving Hamiltonian is free and the quantum state resides in the full  {corresponding} 
Hilbert space when determining the QSL. 
That is, any kind of bi- or multi-partition of the given state is not allowed to affect the QSL.  
Therefore, and importantly for our present argument, isolating the QSL in the inequality
for the power does not impact assessing the influence of entanglement 
contained in the state $\R$ on the maximum power obtainable.}

\subsection{Scaling with cell number}
{Under the condition \eqref{stronger_condition_QSL}, both $\sum_{i=1}^N \tr(\hamz_i^2)$ and $\langle\Delta \hamd^2\rangle$ have linear scaling with $N$, so that $\sqrt{\sum_{i=1}^N \tr(\hamz_i^2)\langle\Delta \hamd^2\rangle}$ also has linear scaling with $N$, which is equivalent to classical scaling for the power without quantum state advantage.
}

{When the energy gap and quantum speed limit are given, $\Gamma_C(\R)$ {therefore} indicates a true quantum advantage,  
related to a geometric advantage in the Hilbert space {related to} entanglement, 
which is, in particular, isolated from the QSL.

Entanglement is in fact necessary (but not sufficient, also see below) to exceed the classical linear 
scaling of the power in the number of cells $N$. One can see this as follows. 
As the quantum state advantage can be calculated by $\Gamma_C(\R)$, we can anticipate the maximum power of a quantum battery. To that end, we use that 
$\sqrt{\Gamma_C(\R)}$ is bounded by $\sqrt{N}$ (Proposition 2 in Appendix \ref{propcommmatrix}), 
and is less than or equal to unity for separable states which {by definition} 
do not have entanglement. This indicates that the maximum quantum state advantage 
 is $\sqrt{N}$ and there is no quantum advantage without entanglement.
This is in agreement with the result for the maximum quantum advantage obtained by \cite{PhysRevLett.118.150601}. 
We can thus indeed quantify the genuine quantum state 
advantage for all states by  $\sqrt{\Gamma_C(\R)}$.

\subsection{Achievability of bound}
The upper bound contained in \eqref{inequal_power} is achievable for all kinds of states by manipulating $\hamz$ and $\hamd$. 
On the other hand,  the power bound can not be saturated for {\em given} $\hamd$ and $\hamz$. 
One of the factors leading to a nonsaturated bound is a deviation of $\hamd$ from the optimal driving Hamiltonian: 
Imposing a restriction on $\hamd$ 
leads to a reduced power bound, which can be quantified by 
\[\label{eq:theta_H}
\cos{\theta_{\hamd}}=-\frac{\tr([i[\hamz,\sqrt{\R}]\{\sqrt{\R},\hamd\})^2}{\tr([\hamz,\sqrt{\R}]^2)\tr(\{\sqrt{\R},\hamd\}^2)}. 
\] 
The structure of $\hamz$ also affects the power bound, since commonly $\hamz$ is assumed to be 
fixed and not controllable, in distinction to the driving Hamiltonian $\hamd$. 
This effect is expressed by another angle 
\[\label{eq:theta_V}
\cos{\theta_{\hamz}}=\frac{-\tr([\hamz,\sqrt{\R}]^2)}{\Gamma_C(\R)\sum_{i=1}^N\tr(\hamz_i^2)}.
\]

The two angles defined as in the above yield an equation for the power, cf.~Appendix.~\ref{Ap:eq_power}, 
\[\label{power_eq} 
P= \sqrt{2\kappa\Gamma_C(\R)\sum_{i=1}^N\tr(\hamz_i^2)\mean{\Delta\hamd^2}\cos{\theta_{\hamd}}\cos{\theta_{\hamz}}}.   
\]
From 
this expression, 
we note that the driving Hamiltonian $\hamd$ 
should increase $\Gamma_C$ during time evolution. 
 Driving Hamiltonians constructed from simple sums of local operators 
 can however not increase entanglement, as dictated by the condition 
 that every conceivable entanglement measure should not increase by local operations and classical communication (LOCC) \cite{Nielsen,LOCC,Guifre}.
This suggests that states evolved by such  driving Hamiltonians 
 starting from the ground state can not have $\Gamma_C$ exceeding unity:  
There is by definition no quantum state advantage from sums of local  driving Hamiltonians. 


\section{Examples for quantum state advantage and disadvantage}
We study below in detail the properties of $\Gamma_C$ for several examples. 
We will find that $\Gamma_C$ is not an entanglement monotone. This fact can lead to a 
counterintuitive behavior of the (supposedly)  quantum origin of charging power.

To quantify the entanglement in a given state $\R$, we employ the negativity, demonstrated to yield an
entanglement monotone by \cite{VidalWerner}. 
For a bipartite system composed, say, of partitions $A$ and $B$, 
the definition of negativity, ${\mathcal N}(\R)$, employs
 $\R^\mathrm{T_A}$, the partial transpose operation 
 on an $A$-dimensional part of the full density matrix $\R$, 
\[
{\mathcal N}=\frac{\|\R^\mathrm{T_A}\|_{\mathrm{Tr}}-1}{2}. 
\]
 Here, $\| \R^\mathrm{T_A}\|_{\mathrm{Tr}}$ is the trace norm, which 
is the sum of the absolute values of the eigenvalues of $\R^\mathrm{T_A}$.


\subsection{Initial GHZ state and thermalization} \label{GHZevol}

We first argue that generally the quantum state advantage $\Gamma_C(\R)$ satisfies 
the following
 \begin{theorem}
 {$\Gamma_C$ is not an entanglement monotone.}\label{theorem2}
 \end{theorem}
 \begin{proof}
 We prove the theorem by a counterexample. 
In particular, LOCC operations can lead to an increase of $\Gamma_C$. Suppose the initial state is prepared to be  $\ket{\psi_i}=\frac{1}{\sqrt{3}}(\ket{00}+\ket{11}+\ket{22})$, which has $\Gamma_C(\ket{\psi_i}\bra{\psi_i})=4/3$. It evolves by LOCC to the final state, a GHZ state 
 $\ket{\psi_f}=\frac{1}{\sqrt{2}}(\ket{00}+\ket{22})$ which has $\Gamma_C(\ket{\psi_f}\bra{\psi_f})=2$, that is 
 larger than $\Gamma_C$ than the initial state. 
 Formally, $\Gamma_C(\ket{\psi_i}\bra{\psi_i})<\Gamma_C(\mathrm{LOCC}(\ket{\psi_i}\bra{\psi_i}))$.
 \end{proof}
 
The entanglement non-monotonicity of $\Gamma_C$ will now be shown to yield a counterintuitive result: Increasing the quantum state advantage during losing entanglement by thermalization. To explicitly describe this phenomenon, we assume a Lindbladian evolution~\cite{10.1093/acprof:oso/9780199213900.001.0001}.   

To simplify the Lindblad equation for our purpose, we assume that each battery cell interacts with its own thermal bath, and the individual baths are uncorrelated among each other, to ensure that spurious entanglement breaking between the battery qudits generated by coupling to a common bath is not taking place. Finally, we assume that the qudit cells do not interact with each other. Under these assumptions, the Lindbladian evolution of the system is equivalent to that of a single cell in a given bath.
Putting a possible Lamb-type shift to zero, we have, for two cells 
\begin{multline}
\label{eq_Lind}
    \frac{\partial\R}{\partial t}=i[\R,\hamz]  
    +g\sum_{i=1,2} \left\{ (N_p(\omega_0)(\hat{a}_i^\dagger\R\hat{a}_i-\{\R,\hat{a}_i\hat{a}_i^{\dagger}\}/2) \right.  \\
    \left. +(N_p(\omega_0)+1)(\hat{a}_i\R\hat{a}_i^\dagger-\{\R,\hat{a}_i^{\dagger}\hat{a}_i\}/2)\right\},
\end{multline}
where the coefficient $g$ is determined by the strength of 
interaction between the qudit cell and the bath. 
Assuming an electromagnetic bath, the number of photons with energy $\omega$ is given by $N_p(\omega)=e^{-\beta \omega}/(1-e^{-\beta \omega})$. The battery Hamiltonian, which reads 
$\hamz=\hat{I}\otimes \sum_{d=1}^D d \omega_0\ket{d}\bra{d}+\sum_{d=1}^D d \omega_0\ket{d}\bra{d} \otimes \hat{I}$, is composed of decoupled cells with a uniform energy gap, $\omega_0$, and  $D$ is the qudit dimension of the cell. 

 \begin{figure}[t!]
    \includegraphics[width=\hsize]{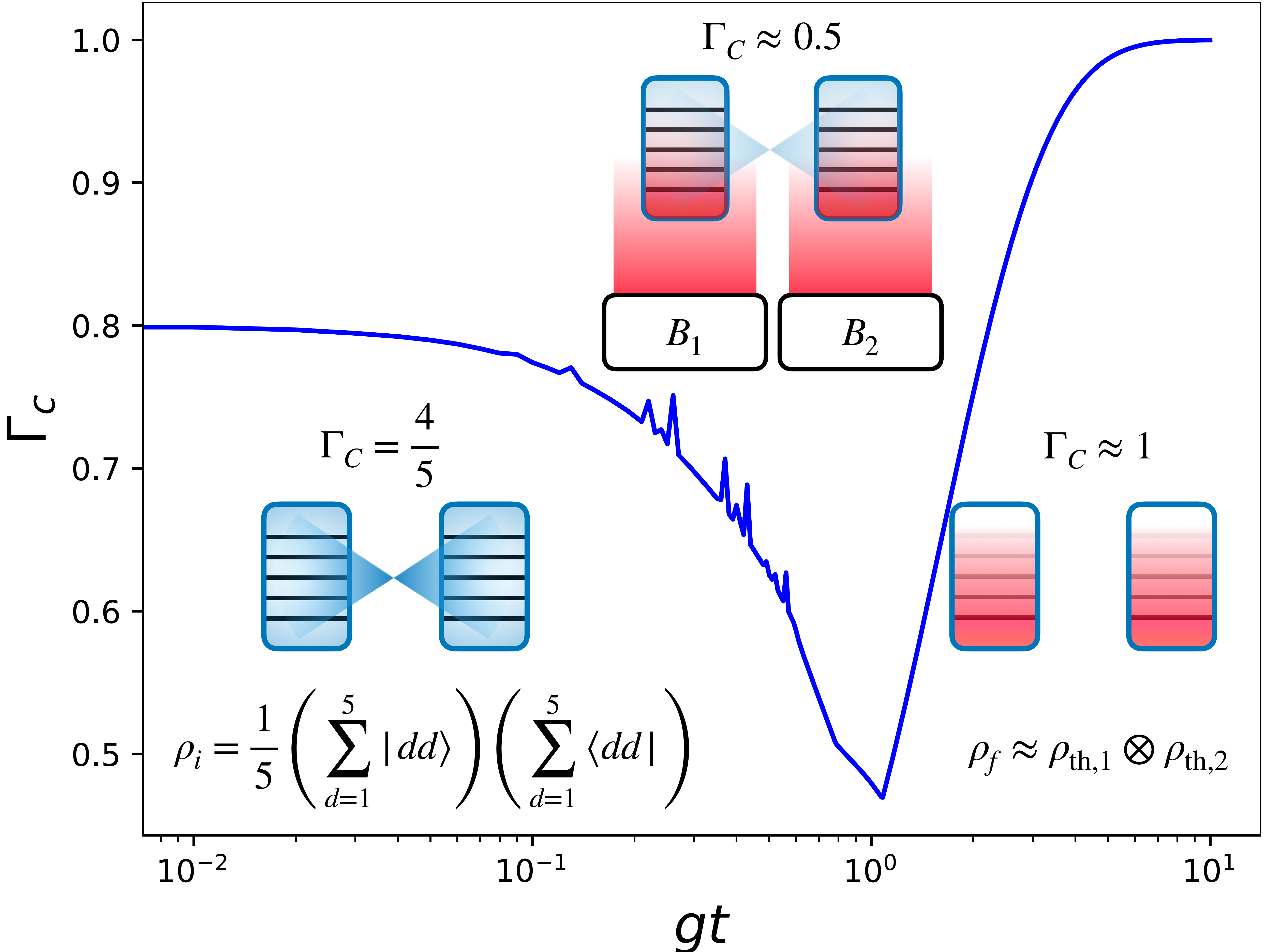} 
    \caption{{Numerically simulated 
    temporal evolution of the quantum state advantage $\Gamma_C(t)$ 
    during thermalization, for two baths $B_1$ and $B_2$ independently coupling to the qudits, for $D=5$.  
    The insets schematically illustrate the  loss of entanglement between the two qudit batteries during the thermalization  process (indicated by the weakening of wedge shading between betteries), for various qudit dimensions, cf.~Fig.~\ref{fig:figure2}. }}
    \label{fancy}
\end{figure}

Initially, at $t=0$, generalized GHZ states are prepared in the two qudits, $\ket{\psi_0}=\frac{1}{\sqrt{D}}\sum_{d=1}^D\ket{dd}$. They represent maximally entangled states, where $\ket{dd}=\ket{d}\otimes\ket{d}$. The quantum state  advantage $\Gamma_C(\ket{\psi_0} 
=4/D$ is 
decreasing with $D$, while the entanglement measure negativity increases with $D$ as ${\mathcal N}=0.25(D-1)$. 

We illustrate in Fig.~\ref{fancy} that during the thermalization process, with the ensuing loss of entanglement,
the quantum state advantage goes through a minimum, but eventually settles at a larger value of $\Gamma_C$ than the initial one for the final, completely thermalized and disentangled state
{for large qudit dimension $D$}.{  

\begin{figure}[t!]
\begin{center}
    \includegraphics[width=0.9\hsize]{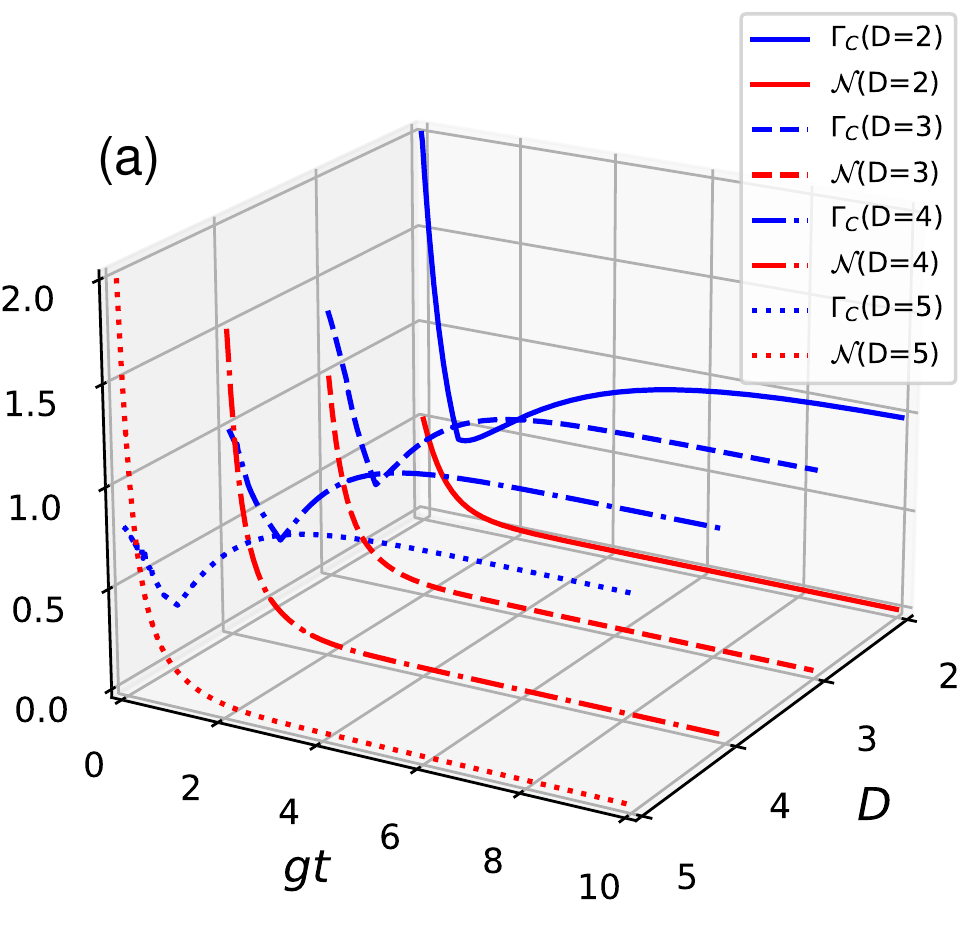}
    \includegraphics[width=0.9\hsize]{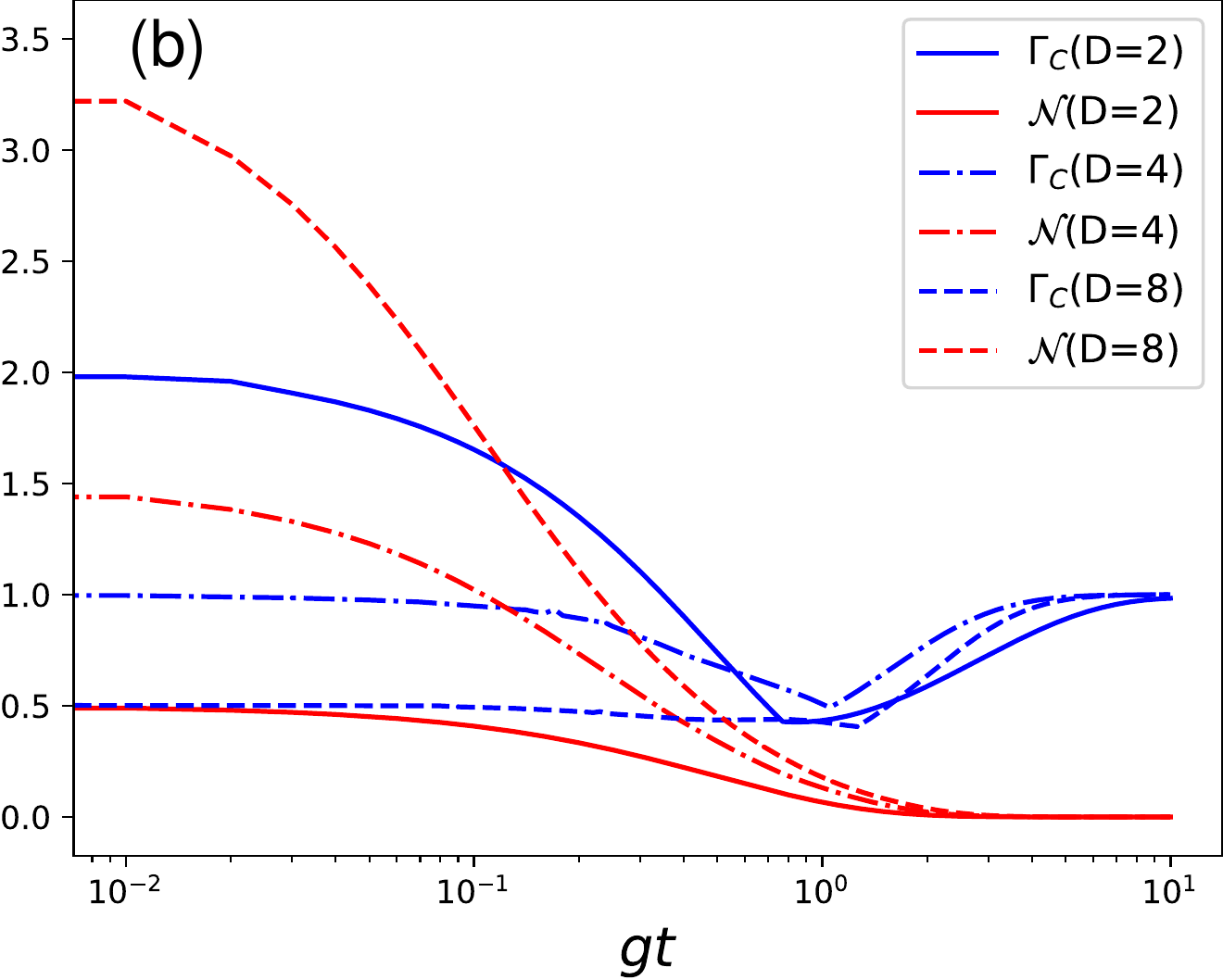}   
    \caption{(a) The time evolution of $\Gamma_C$ (red) and negativity (red), ${\mathcal N}$, via thermalization, for different  qudit dimensions $2\leq D \leq 5$, starting from a generalized GHZ state. 
The thermalization process follows \eqref{eq_Lind}. The temperature of the photonic bath is chosen to be small, $k_B T=0.1$ relative to the energy gap which we set $\omega_0=1$. 
    While negativity monotonically decreases with time  $\forall\, D$, 
    the quantum state advantage 
  $\Gamma_C$ has a turning 
    point and saturates to unity, and is greater than the initial $\Gamma_C=0.8$ {(for $D= 5$).]  
    (b) Logarithmic representation of time evolution for the quantities displayed in the inset for $D=2,4,8$.  
    Note that the speed of the initial decrease of $\Gamma_C$ reduces with qudit dimension 
    and that for $D>4$ the initial $\Gamma_C$ is smaller than its final value.} }
\label{fig:figure2}
\end{center}
\end{figure}

We display  the evolution of $\Gamma_C$ and $\mathcal N$ via  the thermalization process by the Lindbladian evolution~\eqref{eq_Lind}, for various qudit dimensions
in Fig.~\ref{fig:figure2}. 
We observe that negativity and thus the entanglement monotonically decrease with time. 
On the other hand, $\Gamma_C$ has turning points where $\Gamma_C$ starts to increase and then 
saturates to unity. 
This increasing behavior is observed at all dimensions 
$D$. 
Indeed, $\Gamma_C$ can be small even though large entanglement exists, which suggests that 
certain kinds of entanglement can deteriorate the {maximum} quantum battery power. 
We coin this phenomenon {\em entanglement disadvantage}. 
Indeed, for the present example, when states lose their initially large entanglement quantum disadvantage 
by thermalization, $\Gamma_C$ recuperates to its semiclassical value.

As a conclusion, maximally bipartite-entangled states do not necessarily enhance the quantum charging power, {but can potentially act in a detrimental way.} 

\subsection{Generalized W states}
{
To obtain more concrete examples of the potentially low quantum state advantage 
 of highly entangled states, we use states for which $\Gamma_C$ can be calculated analytically. We choose generalized  W states \cite{Wolfgang}, which are commonly referred to as being globally entangled. 
The W states 
for $N$ qubits are given by 
\[
\ket{W_N}=\frac{1}{\sqrt{N}}(\ket{100\cdots0}+\ket{010\cdots0}+\cdots+\ket{000\cdots1}).
\]
As shown in Appendix  \ref{chargadv}, we have the exact result 
$\Gamma_C(\ket{W_N})=(3N-2)/N$ [Eq.~\eqref{Gamma_W_N}], 
which approaches three when $N$ goes to infinity. 
This clearly suggests that {also} global entanglement (entanglement connecting all cells) is not a sufficient condition for a quantum-enhanced scaling advantage with the number of cells $N$ (in the large $N$ limit).
}

\subsection{$k$-local entanglement}
\label{k-local}
{We reiterate that although global entanglement is not sufficient for $N$-scaling of the quantum state advantage, global entanglement is still necessary to outperform classical states, i.e., to enhance charging performance. For $k$-local entanglement, which means that each given cell entangles only with at most $k-1$ other cells, we can divide the set of battery cells into  {partitions} composed of $k$ cells which entangle with each other. In this system, $\Gamma_C\leq k$ generally, see Eq.~\eqref{Gamma_C_k} in Appendix \ref{chargadv}, which indicates that there is no quantum state advantage 
 in the sense of a scaling with $N$. As a result, such a quantum state advantage can only stem from a global operator which generates global entanglement, cf.~\cite{PhysRevLett.128.140501}.
For further discussion see Appendix \ref{chargadv}.
}



{
\subsection{{Sachdev-Ye-Kitaev battery charging}}
{A paradigmatic example which has been explicitly demonstrated in Ref.~\cite{Dario} 
to yield a quantum advantage, that is, a power scaling superextensively (faster than linear in $N$),  
is a charging Hamiltonian of the Sachdev-Ye-Kitaev (SYK) type: Spinless fermions on a lattice, 
with random all-to-all interactions (see for a review \cite{SYKReview}).}

{The battery Hamiltonian is here assumed to be a simple 
magnetic field Hamiltonian 
\[
\hamz=\sum_{i=1}^N \mathfrak{h} \sigma^y_i, \label{hamzSYK}
\]
where $\sigma^y$ is the usual Pauli matrix and $\mathfrak h$ is a ``magnetic field" with units of energy. 
This battery is charged by the SYK Hamiltonian,
\[
\hamd=\sum_{i,j,k,l} J_{ijkl} \hat c^\dagger_i \hat c^\dagger_j \hat c_k \hat c_l,
\]
where $\hat c_j^\dagger$ and $\hat c_j$ are creation and annihilation operators of spinless 
fermions at a given site $j$. 
By the Jordan-Wigner transformation, the $\hat c_j$ can be represented by Pauli operators as  
follows, $\hat c_j=(\sigma^x_j-i\sigma^y_j)\prod_{k=1}^{j-1}\sigma^z_k$, 
{which shows how the SYK charger acts on the battery cells in the 
$\hamz$ of Eq.~\eqref{hamzSYK}.}  
The interaction coefficients $J_{ijkl}$ are complex random variables, 
sampled from a Gaussian distribution with zero mean, 
$f(|J_{ijkl}|) = 
\sqrt{\bar{J}^2/2\pi N^3}e^{-N^3|J_{ijkl}|^2/\bar{J}^2}$. 
The standard deviation reads 
$\bar{J}/N^{3/2}$, where the $N^{3/2}$ from the normalization guarantees an extensive, that is 
linear in $N$ scaling of the energy associated to the charging $\hamd$.}

{The initial state 
is prepared to be the ground state of $\hamz$, $\bigotimes_{i=1}^N\ket{-Y}_i$, with 
energy $-\mathfrak{h}N$; subsequently, this state is evolved by $\hamd$ to a charged state of higher energy.
We simulate the corresponding charging process of the SYK battery, to calculate the 
instantaneous power. 
To this end, we define a  dimensionless time $\tau=t\bar{J}$, and 
a dimensionless power 
\[
\tilde P(\tau)
\coloneqq \ \frac{1}{\sqrt{4\mean{\hamd^2} N\mathfrak{h}^2} }
P\big(\tau 
\big).
\]
By our theorem~\ref{theorem_main}, the dimensionless power $\tilde P$  
is in the present case {bounded precisely by the square root of the 
quantum state advantage, $\sqrt{\Gamma_C}$}. {Classical charging, that is by definition charging 
without generating entanglement, 
has $\sqrt{\Gamma_C}\coloneqq 1$, which thus represents the classical power 
bound at any time.} }

\begin{figure}[t]
\begin{center}
    \includegraphics[width=\hsize]{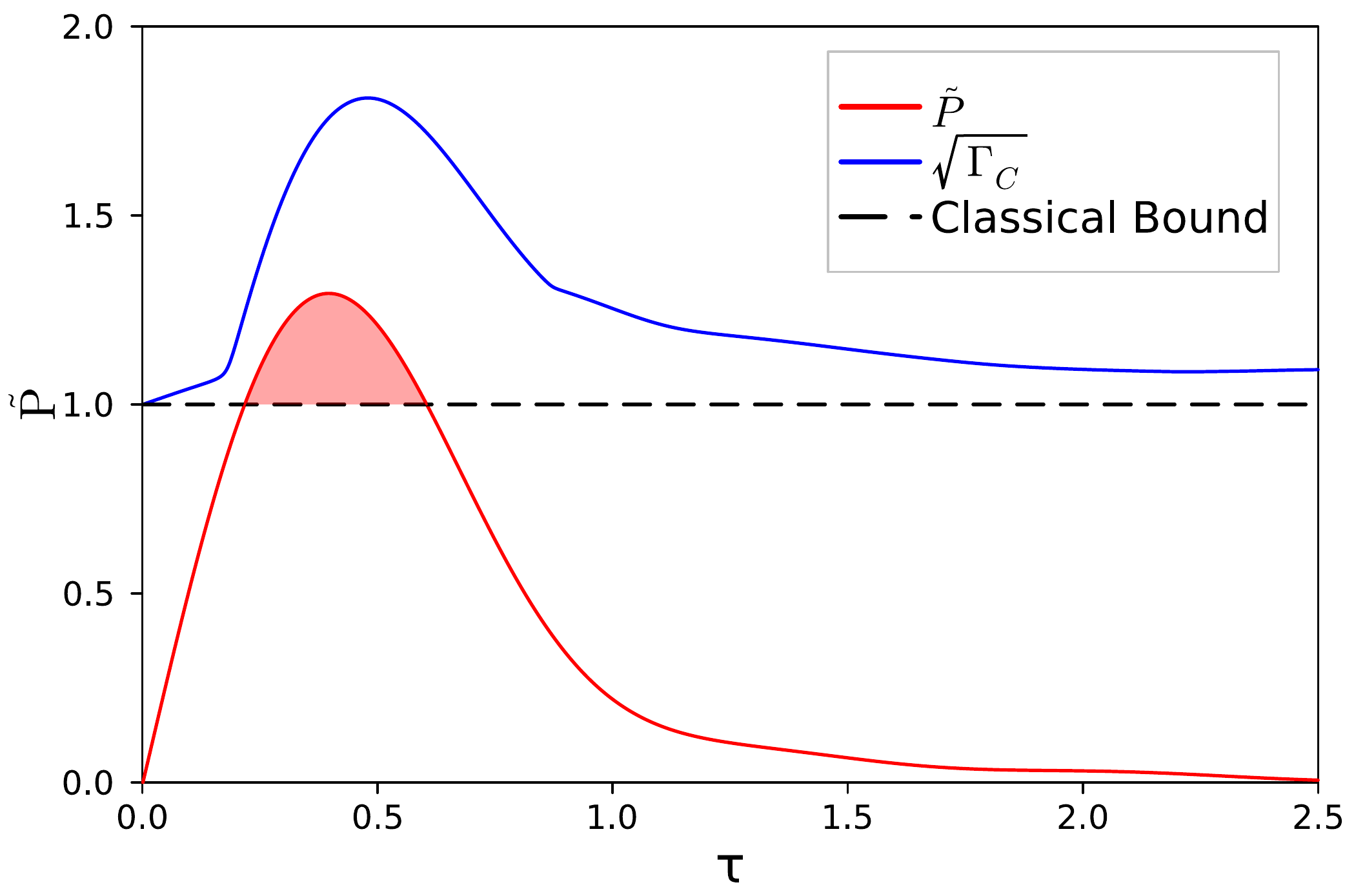}
    \caption{{{\it SYK charging.} Scaled instantaneous power $\tilde P$ as a function of 
    dimensionless time $\tau$  
     and quantum (blue solid) as well as classical (black dashed) power 
     bounds,  for $N=12$ cells. The plot shows one realization 
     of the randomized couplings $J_{ijkl}$; the results for other samplings are very similar.  
     The dashed line shows the classical charging bound corresponding to $\Gamma_C=1$.\label{fig:SYK_plot_qa} }}
\end{center}
\end{figure}

{Fig.~\ref{fig:SYK_plot_qa} plots for one disorder realization 
the dimensionless power $\tilde P$, and dimensionless power 
bound of our system $\sqrt{\Gamma_C}$, 
both as a function of dimensionless time $\tau$. We verified that 
other disorder realizations give very similar results, with 
     $\tilde P(\tau)$ peaking before its bound $\sqrt{\Gamma_C(\tau)}$.}

{We infer from Fig.~\ref{fig:SYK_plot_qa} that while our power bound is not saturated, 
the instantaneous power exceeds the classical bound for all times $\tau$ corresponding to the 
light-red-shaded area in Fig.~\ref{fig:SYK_plot_qa}. {In this time window, we thus obtain a 
genuine quantum state advantage from SYK charging.} }

{Furthermore, we observe from Fig.~\ref{fig:SYK_plot_qa} 
that the quantum state advantage $\Gamma_C(\R)$  
for a SYK-charged battery is not monotonously increasing with time. 
The scaled power $\tilde P$ returns to unity when the battery quantum state becomes 
chaotic due to the random couplings of the SYK charger which does maximally
entangle the cells 
\cite{Dario}. We prove in Appendix~\ref{fchaos} that a maximally entangled final  
state has indeed $\Gamma_C=1$,  
and thus no quantum state advantage \footnote{{We conjecture in this regard that the asymptotics 
of the bound (solid blue line in Fig.~\ref{fig:SYK_plot_qa}) apparently not reaching unity is a finite $N$ effect 
($N=12$ in our simulations).}}.  
}

{This phenomenon of non-monotonicity of $\Gamma_C$, even though the
SYK charger maximally entangles, 
 is another instance illustrating our theorem \ref{theorem2}, which states 
 that $\Gamma_C$ is not an entanglement monotone.} 

\section{Conclusion}

To address the entanglement dependence of quantum charging power in a most 
general, implementation-independent approach, we have employed our  
novel definition of a {\em commutation matrix} instead of the conventionally used covariance matrix. This enables us to isolate the quantum speed limit and the battery energy gap in the bound on the charging power,  from the quantum state advantage $\Gamma_C(\R)$ 
which depends only on the quantum state as represented by the density matrix encoding the entanglement. 
We have thereby, in particular, shown that the maximum obtainable power is not an 
entanglement monotone.
{Note that we have also shown (cf.~Proposition \ref{pro:sap_states} in Appendix \ref{propcommmatrix}),  
that mixtures of product states (separable states), in general 
leading to quantum correlations beyond entanglement such as for example measured by quantum discord \cite{Zurek,Adesso_2016},  
do not imply a quantum state advantage in the sense we have put forward here. That is, they do not lead to 
a quantum advantage which is isolable from speed limit and gap. }

 
Indeed, the instantaneous power of battery charging is obviously dependent 
on the intermediate states between the initial and final states of the charging protocol. 
Our derived bound then demonstrates the necessity of the entanglement of any intermediate 
states to gain an instantaneous quantum advantage of battery charging. 
This statement is, in particular, independent of the preparation of the initial state.
We however also again stress that while entanglement is a necessary condition for a quantum state advantage 
 $\Gamma_C > 1$, we have argued that there also exists a potential quantum {disadvantage}, that is,  
highly entangled states can have a potentially lower maximum charging power than 
{separable states}, and entanglement is, while necessary, not sufficient to obtain a quantum state advantage. 

While we have provided concrete calculations for negativity as a measure of entanglement, 
we anticipate that the qualitative behavior of the quantum state advantage measure $\Gamma_C(\R)$ 
remains unaffected for other measures of bipartite entanglement.  
Finally, the generalization to multipartite entanglement measures, and to develop a general classification scheme for beneficially or detrimentally entangled states are left for future study. 


\section*{Acknowledgments} 
JYG thanks 
JungYun Han for discussions and 
Jeongrak Son for help with the numerics. This work has been supported by the National Research Foundation of Korea under 
Grants No.~2017R1A2A2A05001422 and No.~2020R1A2C2008103.

\section*{Conflict of Interest}

The authors have no conflicts to disclose.

\section*{Data Availability} 

The data that support the findings of this study are available from the corresponding author upon reasonable request.

\begin{appendix}

\section{Several properties of the commutation matrix}
\label{propcommmatrix}
\begin{proposition}
    $\gamma_C(\R,\hat{\bm M})$ is a positive semidefinite matrix for any $\R$. 
\end{proposition}
\begin{proof}{
    For any vector ${\bm u}$, we can define a corresponding operator $\hat O=\hat{\bm M}\cdot {\bm u}$. By the definition of $\gamma_C(\R,\hat{\bm M})$ in Eq.~\eqref{defineGP}, we obtain 
    \[\begin{split}
        &{\bm u}^\mathrm{T}\gamma_C(\R,\hat{\bm M}){\bm u}=\sum_{\mu,\nu}u_\mu u_\nu\gamma_C(\R,\hat{\bm M})_{\mu\nu}\\
       &=-\sum_{\mu,\nu}u_\mu u_\nu\tr([\hat{M_\mu},\sqrt{\R}][\hat{M_\nu},\sqrt{\R}])\\
        &=-\tr([\sum_{\mu}u_\mu\hat{M_\mu},\sqrt{\R}][\sum_{\nu}u_\nu\hat{M_\nu},\sqrt{\R}])\\
        &=-\tr([\sqrt{\R},\hat O]^2).
    \end{split}
    \]}    
Both  $\R$ and $\hat O$ are Hermitian operators; then $[\sqrt{\R},\hat O]$ is an anti-Hermitian operator. Since the eigenvalues of $[\sqrt{\R},\hat O]$ are imaginary, $[\sqrt{\R},\hat O]^2$ is a negative semidefinite matrix.
    Hence $-\tr([\sqrt{\R},\hat O]^2)$ is zero or positive for any $\hat O$ corresponding to a given ${\bm u}$, which proves the proposition.
\end{proof}

\begin{proposition} \label{pro:max_G}
The quantum state advantage  $\Gamma_C(\R)$ is bounded by the number of cells,
    \[
    \Gamma_C(\R)\leq N
    \]
\end{proposition}
\begin{proof} 
    This will be proved by the fact that our inequality~\eqref{inequal_power} is always achievable, and by using Eq.~\eqref{power_uncertanty}.

    There always exist some $\hamz$ and $\hamd$ which satisfy the equality conditions~Eqs.~\eqref{eq_con_1} and \eqref{eq_con_2}.
    By Eq.~\eqref{power_uncertanty}, we obtain
    \[
    \begin{split}
        \sqrt{2\kappa \Gamma_C(\R)} =\frac{P}{\sqrt{\sum_{i=1}^N\tr(\hamz_i^2)\mean{\Delta \hamd^2}}}\\
        \leq2\sqrt{\|\gamma(\R,\hat{\bm M})\|}\leq\sqrt{2N}. 
    \end{split}
    \]
    Since $1\le\kappa\le 2$, the proposition is proven.
\end{proof}

\begin{proposition} \label{pro:sap_states}
{For separable states, a linear summation of product states, }
    $\Gamma_C(\R)$ is less than or equal to unity; 
    there is no quantum advantage for any such states.
\end{proposition}
\begin{proof}
    The proof is similar to the proof of proposition~\ref{pro:max_G}, but here we will use a decomposition of the separable state $\R$ as follows.

    By the definition of separable states, $\R$ can  be decomposed {by a finite number of product states} as
    \[
    \R=\sum_l p_l\R_l. 
    \]
    In this expression, $\R_l$ can be completely expressed {by products of states} and $\sum_l p_l=1$. 
    We can then obtain a bound on the power by Eq.~\eqref{power_uncertanty}, 
    \[\label{ineq:sep_states}
    P\leq 2\sum_l p_l \sqrt{\|\gamma(\R_l,\hat{\bm M})\|\sum_{i=1}^N\tr(\hamz_i^2)\mean{\Delta V^2}_l},
    \]
    where any $\|\gamma(\R_l,\hat{\bm M})\|$ is equal to 1/2 by definition of the covariance matrix. Using the inequality 
    $\sqrt{\mean{\Delta \hamd^2}}\geq \sum_l p_l\sqrt{\mean{\Delta \hamd^2}_l}$,  
    there always exist $\hamz$ and $\hamd$ which satisfy the equality conditions in \eqref{eq_con_1} and \eqref{eq_con_2}.
    With Eq.~\eqref{power_uncertanty}, Eq~\eqref{ineq:sep_states}] is rewritten as
    \[
    P\leq \sqrt{2\sum_{i=1}^N\tr(\hamz_i^2)\mean{\Delta V^2}},
    \]
    which yields 
    \[
        \sqrt{2\kappa \Gamma_C(\R)} =\frac{P}{\sqrt{\sum_{i=1}^N\tr(\hamz_i^2)\mean{\Delta \hamd^2}}}
        \leq\sqrt{2} 
    \]
    Since $1\le\kappa\le 2$, the proposition is proven.
\end{proof}

\begin{proposition}
    The eigenvalues of $\gamma_C(\R)$ are conserved 
    {during evolution with local driving Hamiltonians.} 
\end{proposition}
\begin{proof}
A local driving Hamiltonian is given as 
    \[
    \hamd=\sum_{i=1}^N \hamd_i\bigotimes_{j\neq i}\hat{I}_j,
    \]
    such that a unitary operator from $\hamd$ is expressed by 
    \[
    \hat U =e^{-i\hamd T}=\prod_i^N e^{-i\hamd_i\bigotimes_{j\neq i}\hat{I}_j}=\prod_i^N \hat U_i.
    \label{decomposeU}
    \]
    We then obtain the $\gamma_C$ matrix for the time-evolved state $\hat \hat U\R \hat \hat U^\dagger$ 
    as follows 
    \[
    \begin{split}
    \gamma_C( \hat{U}\R \hat{U}^\dagger,\hat{\bm M}_{\mu\nu})=-\tr([\sqrt{ \hat{U}\R \hat {U}^\dagger},\hat{M}_\mu][\sqrt{\hat{U}\R \hat{U}^\dagger},\hat{M}_\nu])\\
    =-\tr([\sqrt{\R},\hat{U}^\dagger\hat{M}_\mu \hat{U}][\sqrt{\R },\hat{U}^\dagger\hat{M}_\nu \hat{U}]).
    \end{split}
    \]
    Without loss of generality, let us assume that $\hat{M}_\mu$ derives from the $i$th cell. 
    Then, $\hat{M}_\mu$ can be decomposed as $\bigotimes_{ j\neq i}\hat{I}^j \otimes A^i_\mu$. By also using the decomposition of $\hat U$ 
    into the $\hat U_i$ according to \eqref{decomposeU}, we can write 
    \[
    \hat U\bigotimes_{ j\neq i}\hat{I}^j \otimes A^i_\mu \hat U^\dagger=\bigotimes_{ j\neq i}\hat{I}^j \otimes \hat U_i A^i_\mu \hat U_i^\dagger, 
        \]
{due to the orthonormality condition \eqref{Acond}, $U_i A^i_\mu \hat U_i^\dagger=\sum_\xi O_{\mu\xi}A^i_\xi$ for an orthogonal matrix $O$.}
    
     We can expand the orthogonal matrix $O$ using the full space spanned by the $\hat{\bm M}$, which leads to $U^\dagger\hat{M}_\mu U=\sum_\xi O_{\mu\xi}\hat{M}_\xi$. 
     Hence we can write 
     \[
     \gamma_C(U\R U^\dagger,\hat{\bm M})=O^\mathrm{T}\gamma_C(\R,\hat{\bm M})O.
     \]
    {Therefore the local driving Hamiltonians only act on the state 
     by rotating the basis of the commutation matrix $\gamma_C$, 
     and do not change its 
    eigenvalues and hence also its norm, so that 
     $\Gamma_C$ is not affected.}
     \end{proof}

 \section{{Proof of inequality \eqref{inequal_power}}} 
 \label{proof_power} 
In this Appendix, we will expand on the detailed derivation of Theorem~\ref{theorem_main}, for which 
the primary steps were outlined in the main text. Let us start with establishing a Cauchy–Schwarz type  inequality for operators.

For Hermitian operators $\hat{A}=\sum_{i}a_i\ket{a_i}\bra{a_i}$ and $\hat{B}=\sum_{i}b_i\ket{b_i}\bra{b_i}$,  the 
following Cauchy-Schwarz type inequality holds:
\bea
|\tr(\hat{A}\hat{B})|\leq\sqrt{\tr(\hat{A}^2)\tr(\hat{B}^2)}\label{CauchyIn}
\ea
which is derived by an arithmetic geometric-mean inequality as follows 
\bea
\begin{split}
& 2
=\sum_{ij}\left(\frac{a_i^2}{\tr(\hat{A}^2)}+\frac{b_i^2}{\tr(\hat{B}^2)}\right)|\braket{a_i}{b_j}|^2\\
& \geq\sum_{ij} \frac{2|a_i b_j|}{\sqrt{\tr(\hat{A}^2)\tr(\hat{B}^2)}} |\braket{a_i}{b_j}|^2\geq2\frac{|\tr(\hat{A}\hat{B})|}{\sqrt{\tr(\hat{A}^2)\tr(\hat{B}^2)}}.
\end{split}
\ea

By replacing $\hat{A}$ and $\hat{B}$ in \eqref{CauchyIn} with $i[\hat{H},\sqrt{\R}]$ and $\{\sqrt{\R},\hat{V}\}$, we obtain the following bound on $P$, 
\bea
|P|&=&|\tr(i[\hat{H},\sqrt{\R}]\{\sqrt{\R},\hat{V}\})|\nonumber \\
&\leq&\sqrt{\tr(-[\hat{H},\sqrt{\R}]^2)\tr(\{\sqrt{\R},\hat{V}\}^2)},
\label{CuachyP}
\ea
which implies that $\tr(-[\R,\hat{H}]^2)$ determines the quantum advantage coming 
from the state $\R$ since we isolated the QSL contribution, encoded in $\{\sqrt{\R},\hat{V}\}^2$, 
on the right-hand side of the power inequality, cf.~the discussion below Eq.~\eqref{inequal_power_lambda}. 

We now use that the battery Hamiltonian can be written as$\hamz=\sqrt{\sum_{i=1}^N\tr(\hamz_i^2)} \hat{\bm M}\cdot {\bm u}$
which gives the equality from the definition of $\gamma_C(\R,\hat{\bm M})$
\bea
-\tr([\R,\hat{H}]^2)=\sum_{i=1}^N\tr(\hamz_i^2)u^\mathrm{T}\gamma_C(\R,\hat{\bm M})u.
\ea
By using the property of the operator norm that it is equal to the largest
eigenvalue of the operator, we have the inequality 
\bea
-\tr([\R,\hat{H}]^2)\leq\norm{\gamma_C(\R,\hat{\bm M})}\sum_{i=1}^N\tr(H_i^2). \label{fyou}
\ea
The quantity $\hat{V}+\lambda \hat{I}$ has identical power $P\,\,\forall\,\lambda$ and $\tr(\{\hamd+\lambda \hat{I},\sqrt{\R}\})$ is minimized by $2\tr(\sqrt{\R}\hamd\sqrt{\R}\hamd)+2\tr(\R\hamd^2)-4\tr(\R\hamd)^2$ when $\lambda=-\tr(\hamd\R)$.

Since $-\tr([\sqrt{\R},\hamd]^2)=2\tr(\sqrt{\R}\hamd\sqrt{\R}\hamd)-2\tr(\R\hamd^2)\geq 0$, we can establish the inequality 
\[
\begin{split}
&2\tr(\sqrt{\R}\hamd\sqrt{\R}\hamd)+2\tr(\R\hamd^2)-4\tr(\R\hamd)^2\\
&\leq 4\tr(\R\hamd^2)-4\tr(\R\hamd)^2=4\mean{\Delta\hamd^2}.
\end{split} \label{IneqV}
\]
Moreover, $\tr(\sqrt{\R}(\hamd+\lambda\hat{I})\sqrt{\R}(\hamd+\lambda\hat{I}))$ is greater or equal to zero 
$\forall\,\lambda$, so we obtain that $\tr(\sqrt{\R}\hamd\sqrt{\R}\hamd)-\tr(\R\hamd)^2\geq 0$. Due to the above inequality \eqref{IneqV}, $\tr(\{\sqrt{\R},\hamd\})$ is then bounded by
\[
\begin{split}
&2\mean{\Delta\hamd^2}=2\tr(\R\hamd^2)-2\tr(\R\hamd)^2\\
&\leq2\tr(\sqrt{\R}\hamd\sqrt{\R}\hamd)+2\tr(\R\hamd^2)-4\tr(\R\hamd)^2.
\end{split}
\]
Now we have $\tr(\{\sqrt{\R},\hamd\}^2)=2\kappa \mean{\Delta\hamd^2}$ {defining a
coefficient $\kappa$ which lies in the range $1\leq\kappa\leq 2$.} 

Consequently, we obtain a power bound in the form 
\[
P\leq \sqrt{2\kappa\Gamma_C(\R)\sum_{i=1}^N\tr(\hamz_i^2)\mean{\Delta\hamd^2}},
\]
the relation \eqref{inequal_power} which we set out to prove. 

Next we will discuss the condition for a saturation of the bound~\eqref{inequal_power}. The equality condition of \eqref{CauchyIn} is $\hat{A}=c\hat{B}$ for an arbitrary real number $c$. Hence the first equality condition~\eqref{eq_con_1} arises. The remaining question is which driving Hamiltonian $\hamd$ satisfies this condition. 

Let us investigate \eqref{eq_con_1} within a density matrix eigenbasis. 
In this basis, we can write 
$\R=\sum_\alpha \rho_\alpha \ket{\alpha}\bra{\alpha}$, $\hamz=\sum_{\alpha \beta} h_{\alpha \beta} \ket{\alpha}\bra{\beta}$ and $\hamd=\sum_{\alpha \beta} v_{\alpha \beta} \ket{\alpha}\bra{\beta}$. Then \eqref{eq_con_1} can be formulated as follows
\[
\label{eq_con1_de}
\sum_{\alpha\beta} v_{\alpha \beta}(\sqrt{\rho_\alpha}+\sqrt{\rho_\beta})\ket{\alpha}\bra{\beta}=ic\sum_{\alpha \beta} h_{\alpha \beta}(\sqrt{\rho_\beta}-\sqrt{\rho_\alpha})\ket{\alpha}\bra{\beta},
\]
for a constant $c=\sqrt{-\tr(\{\sqrt{\R},\hamd\}^2)/\tr([\sqrt{\R},\hamz]^2)}$. The solutions of \eqref{eq_con1_de} 
are given by 
\[
\label{eigen_v_ij}
v_{\alpha\beta}=ic\frac{\sqrt{\rho_\beta}-\sqrt{\rho_\alpha}}{\sqrt{\rho_\alpha}+\sqrt{\rho_\beta}}h_{\alpha\beta} \qquad \forall \, \alpha,\beta .
\]
As a result, there always exist driving $\hamd$ satisfying the equality condition~\eqref{eq_con_1}.

{The second equality condition \eqref{eq_con_2} comes from the definition of operator norm. The driving $\hamz$ is decomposed in terms of the $\hat{\bm M}$ as $\hamz = \sqrt{\sum_i^N\tr(\hamz_i^2)}\hat{\bm M}\cdot{\bm u}$ for real normalized vectors ${\bm u}$. By the definition of $\gamma_C(\R,\hat{\bm M})$, $-\tr([\sqrt{\R},\hamz]^2)=(\sum_i^N\tr(\hamz_i^2)){\bm u}^\mathrm{T}\gamma_C(\R,\hat{\bm M}){\bm u}$. } 
We now use that 
the operator norm is equal to the maximum absolute value of all 
eigenvalues of a symmetric real matrix, 
The commutation matrix 
$\gamma_C(\R,\hat{\bm M})$ is such a real symmetric matrix, 
{so there always exists an optimal direction of ${\bm u}={\bm u}_m$ 
projecting out the largest eigenvalue of $\gamma_C$, which finally gives 
\[
\begin{split}
-\tr([\sqrt{\R},\hamz]^2)&=(\sum_i^N\tr(\hamz_i^2)){\bm u}_m^\mathrm{T}\gamma_C(\R,\hat{\bm M}){\bm u}_m\\
    &=(\sum_i^N\tr(\hamz_i^2))\Gamma_C(\R)
\end{split}
\]}\noindent Now we 
established the second equality condition~\eqref{eq_con_2} and the fact that it is always achievable for any $\R$.

\section{Proof of equation \eqref{power_eq}}\label{Ap:eq_power}

In Appendix~\ref{proof_power}, we proved that our inequality for the 
power can be saturated  
by manipulating the battery driving Hamiltonians, $\hamz$ and 
$\hamd$, respectively. However, we can not access all kinds of Hamiltonians; we only have a 
restricted choice of Hamiltonians, which renders the bound on the power~\eqref{inequal_power} loose.

To quantify how loose our bound actually is, we define two angles $\theta_{\hamz}$ and $\theta_{\hamd}$, by  Eqs.~\eqref{eq:theta_H} and \eqref{eq:theta_V}. By substituting $\theta_{\hamz}$ and $\theta_{\hamd}$ into  Eq.~\eqref{power_eq}, we can confirm it is identically fulfilled. 
 We bring Eq.~\eqref{power_eq} into the form
 \[
 \label{eq:power_in_App}
P=\sqrt{2\kappa\Gamma_C(\R)\sum_{i=1}^N\tr(\hamz_i^2)\mean{\Delta\hamd^2}\cos{\theta_{\hamd}}\cos{\theta_{\hamz}}}.
\]
By the definition of $\theta_{\hamz}$, we obtain that Eq.~\eqref{eq:power_in_App} is equivalent to 
\[\label{eq:power_in_App_2}
P=\sqrt{-2\kappa\tr([\hamz,\sqrt{\R}]^2)\mean{\Delta\hamd^2}\cos{\theta_{\hamd}}}.
\]
Because $\kappa$ is defined by $\tr(\{\sqrt{\R},\hamd\}^2)/2\mean{\Delta\hamd^2}$, the 
Eq.~\eqref{eq:power_in_App_2} can be represented as
\[
P=\sqrt{-\tr([\hamz,\sqrt{\R}]^2)\tr(\{\sqrt{\R},\hamd\}^2)\cos{\theta_{\hamd}}},
\]
{which is, as required, identically true:}
\[
\tr(i[\hamz,\sqrt{\R}]\{\sqrt{\R},\hamd\})=\tr(i\hamz[\R,\hamd])\coloneqq P,
\]
by using the definition of $\theta_{\hamd}$ and the cyclic property of the trace.

    \begin{widetext}
\section{Quantum state advantage $\Gamma_C$ for W and $k$-locally entangled states}
\label{chargadv}
{In the main text, we have provided the quantum advantage of pure states 
for several examples. We detail here the corresponding derivations.}

\subsection{W states}
We first prove for generalized W states that $\Gamma_C(\ket{W_N})$ is the function $\Gamma_C(\ket{W_N})=(3N-1)/N$. 
Since W states only occupy a finite number of $N$ Fock space states, this 
can be analytically derived.
The orthonormal set \eqref{M} 
for noninteracting batteries is given by ($1\leq i,j \leq N$) 
\[
\begin{split}
    \hat{\bm M}= 
    \left\{\frac{1}{\sqrt{2}} \sigma_i^x\bigotimes_{j\neq i}\hat{I}_j,\frac{1}{\sqrt{2}}\sigma_i^y\bigotimes_{j\neq i}\hat{I}_j,\frac{1}{\sqrt{2}}\sigma_i^z\bigotimes_{j\neq i}\hat{I}_j,\frac{1}{\sqrt{2}}\hat{I}_i\bigotimes_{j\neq i}\hat{I}_j\right\}.
\end{split}
\]
We define a set of matrices 
$\gamma^\sigma (1\leq \sigma\leq 4)$ as follows 
\[
\gamma^\sigma_{ij}\vcentcolon  =\gamma_C(\ket{W_N},\hat{\bm M})_{3i+\sigma,3j+\sigma}, 
\]
\noindent where the $i,j=1,\ldots,N$ are site indices as previously. 
We then have  $\gamma_C(\ket{W_N},\hat{\bm M})=\gamma^1 \bigoplus \gamma^2\bigoplus \gamma^3\bigoplus
\gamma^4$, since the $\gamma_C(\ket{W_N},\hat{\bm M})_{3i+k,3j+l}$ vanish $\forall\, i,j$ when 
$k\neq l$.

The four matrices  $\gamma^\sigma$ have the following elements 
\bea
&\gamma^1_{ij}=\gamma^2_{ij}=
&\begin{cases}
1\quad&\mathrm{for\, \, }i=j\\
\frac{2}{N}\quad&\mathrm{for}\,\, i\neq j
\end{cases}, 
\quad \qquad \gamma^3_{ij}=
\begin{cases}
\frac{4N-4}{N^2}\quad&\mathrm{for}\,i=j\\
\frac{4}{N^2}\quad&\mathrm{for}\,i\neq j
\end{cases}, \qquad \quad 
\gamma^4_{ij}=
0,
\ea
\end{widetext}
from which we can readily obtain their eigenvalues. By the definition of $\Gamma_C$, it then follows 
\[
\begin{split}\label{Gamma_W_N} 
    &\Gamma_C(\ket{W_N})=\|\gamma_C(\ket{W_N},\hat{\bm M})\|\\
    &=\max(\|\gamma^1\|,\|\gamma^2\|,\|\gamma^3\|)=\frac{3N-2}{N}.
\end{split}
\]
We have confirmed this analytical result by direct numerical evaluation of $\Gamma_C(\ket{W_N})$.  

\subsection{$k$-locally entangled states}
{As stated in the main text, a $k$-local entangled 
system is divided into partitions in which {at most $k$} 
cells are entangled within each partition, but the {cell states} in different partitions
are not entangled among each other. 
Formally, this assumption {on the locality of entanglement} is {therefore} expressed as
\[
\ket{\psi}=\bigotimes_{p=1}^P\ket{\psi_p},
\]
where all $\ket{\psi_p}$ are composed from $k$ cells at most and $p$ labels all possible 
partitions. } 

{In the present case, $\gamma_C(\psi)$ can be represented by a simple sum 
of the commutation matrices from each partition, 
\[
\gamma_C(\psi)=\bigoplus_{p=1}^P \gamma_C(\psi_p),
\]
because of the given assumption that there is no entanglement between cells from different partitions. 
We finally obtain  
\[\label{Gamma_C_k} 
\Gamma_C(\ket{\psi})=\|\gamma_C(\ket{\psi})\|=\max_{p}\|
\gamma_C(\ket{\psi_p})\|\leq k,
\]
since each $\|\gamma_C(\ket{\psi_p})\|$ is less than or equal to the number of cells in the given partition.
}

\section{{Maximally entangled quantum chaos implies $\Gamma_C=1$}}\label{fchaos}

{The SYK charging operator 
generates many-body chaos with a large amount of entanglement~\cite{PhysRevLett.126.030602}. 
We assume in the following that the final 
chaotic states, {$\ket{\psi_C}$} have maximal entanglement entropy for any partition, such that the entanglement entropy for the partition $A$, which has a 
$d_A$-dimensional Hilbert space, reads }
{\[
-\tr(\R_A\ln(\R_A))=\ln(d_A),
\]
with the reduced density matrix, traced over the complement of $A$
\[
\R_A=\tr_{\neg A}(\pro{\psi_C}{\psi_C}). 
\]
To maximize the entanglement entropy, the mixed state, $\R_A$ should be a maximally mixed state, 
such that
\[
\R_A=\hat{I}/d_A. \label{assumption} 
\]
This implies that a mixed state for a given single cell has density matrix $\R_i=\hat{I}/2$ and a 
mixed state for two cells has $\R_{ij}= \hat{I}/4\,\forall i\neq  j$ since the cells in our case are qubits. }

{By definition $\hat{M}_\mu$ is an element of $\hat{\bm M}^i$ for some $1\leq i\leq N$. Hence, we can assume $\hat{M}_\mu=\hat{A}_\alpha^i\bigotimes_{k\neq i}\hat{I}^k$ and $\hat{M}_\nu=\hat{A}_\beta^j\bigotimes_{k\neq j}\hat{I}^k$ without loss of generality. When $i\neq j$, we obtain that
\[\begin{split}
\gamma_C(\pro{\psi_C}{\psi_C},\hat{\bm M})_{\mu\nu}\!\!=\!\tr([\hat{M_\mu},\pro{\psi_C}{\psi_C}][\hat{M_\nu},\pro{\psi_C}{\psi_C}]),
\end{split}
\]
{because of the pure state property}  $\sqrt{\pro{\psi_C}{\psi_C}}=\pro{\psi_C}{\psi_C}$.
We trace out all sites $k\neq i,j$, 
which yields 
\[\begin{split}\label{eq:gamma_chaos}
\gamma_C(\pro{\psi_C}{\psi_C},\hat{\bm M})_{\mu\nu}\!=2\tr(\hat{A}_\alpha^i\otimes\hat{A}_\beta^j \R_{ij})\\
-2\tr(\hat{A}_\alpha^i\otimes\hat{I}^j \R_{ij})\tr(\hat{I}^i\otimes\hat{A}_\beta^j \R_{ij})
\end{split}
\]
By using our assumption \eqref{assumption}, $\R_{ij}=\hat{I}/4$, 
Eq.~\eqref{eq:gamma_chaos} becomes 
\[\begin{split}
    &\gamma_C(\pro{\psi_C}{\psi_C},\hat{\bm M})_{\mu\nu}\\
    &=\tr(\hat{A}_\alpha^i\otimes\hat{A}_\beta^j)/2-\tr(\hat{A}_\alpha^i\otimes\hat{I}^j )\tr(\hat{I}^i\otimes\hat{A}_\beta^j)/2=0.
\end{split}
\]}

{When $i=j$, we trace out all $k\neq i$ such that
\[\begin{split}
    \gamma_C(\pro{\psi_C}{\psi_C},\hat{\bm M})_{\mu\nu}\!=\!2\tr(\hat{A}_\alpha^i\hat{A}_\beta^i \R_{i})\!-2\tr(\hat{A}_\alpha^i\R_{i})\tr(\hat{A}_\beta^i \R_{i}).
\end{split}
\]
By using the assumption \eqref{assumption}, $\R_i=\hat{I}/2$,
\[
\begin{split}
\label{eq:chaos_end}
        \gamma_C(\pro{\psi_C}{\psi_C},\hat{\bm M})_{\mu\nu}\!=2\tr(\hat{A}_\alpha^i\hat{A}_\beta^i )-\tr(\hat{A}_\alpha^i)\tr(\hat{A}_\beta^i),
\end{split}
\]
which identically vanishes when $\alpha\neq\beta$ by the orthogonality condition \eqref{Acond}  
fulfilled by the $\hat{A}_\alpha$.}

{Now the $\alpha=\beta$ case remains to be assessed. 
For a two-level system, the set of the $\hat{A}_\alpha$ is composed of {normalized 
 Pauli matrices}, $\sigma_i^x/\sqrt{2}, \sigma^y/\sqrt{2},\sigma^z/\sqrt{2}$ and the identity operator $\hat{I}/\sqrt{2}$. When $\hat{A}_\alpha^i$ is a {normalized Pauli matrix}, Eq.~\eqref{eq:chaos_end} yields unity 
and it gives zero for $\hat{A}_\alpha^i=\hat{I}/\sqrt{2}$.}

{Consequently, $\gamma_C(\pro{\psi_C}{\psi_C},\hat{\bm M})_{\mu\nu}$ is a diagonal matrix with elements unity or zero. Hence $\Gamma_C{\pro{\psi_C}{\psi_C}}=\|\gamma_C(\pro{\psi_C}{\psi_C},\hat{\bm M})_{\mu\nu}\|$ is equal to unity, which proves the assertion stated in the section header. }

\end{appendix}
\bibliography{ep22}

\end{document}